\documentclass[letterpaper, 10 pt,conference]{ieeeconf}  

\IEEEoverridecommandlockouts                              

\usepackage{graphicx}      
\graphicspath{{Figures/}}
\usepackage{float,amsmath,amsfonts,xcolor,graphicx,stmaryrd}
\usepackage{mathtools, cuted}
\usepackage[ruled,vlined,linesnumbered]{algorithm2e}
\usepackage{algpseudocode}

\newtheorem{theorem}{Theorem}[section]
\newtheorem{lemma}{Lemma}[section]
\newtheorem{example}{Example}[section]

\newtheorem{definition}{Definition}[section]
\usepackage{xcolor}

\definecolor{wheat}{rgb}{0.96,0.87,0.70}

\DeclareMathOperator*{\argmin}{arg\,min}

\newif\ifconfidential

 \confidentialtrue
 
 \ifconfidential
\usepackage{fancyhdr}
\fi
\graphicspath{{Figures/}}

\title{\LARGE \bf
Robustness Measures and Monitors for Time Window Temporal Logic
}

\author{Ahmad Ahmad$^1$, Cristian-Ioan Vasile$^2$, Roberto Tron$^1$, and Calin Belta$^1$
\thanks{} 
\thanks{$^1$ Ahmad Ahmad, Roberto Tron, and Calin Belta ({\tt\small \{ahmadgh,tron,cbelta\}@bu.edu}) are with the Division of System Engineering, Boston University, Boston, MA 02215, USA. $^2$ Cristian Vasile ({\tt\small cvr519@lehigh.edu}) is with Mechanical Engineering
and Mechanics Department at Lehigh
University, Bethlehem, PA, 18015, USA }%
}

\begin{document}

\maketitle
\thispagestyle{empty}
\pagestyle{empty}

\begin{abstract}
Temporal logics (TLs) have been widely used to formalize interpretable tasks for cyber-physical systems. Time Window Temporal Logic (TWTL) has been recently proposed as a specification language for dynamical systems. In particular, it can easily express robotic tasks, and it allows for efficient, automata-based verification and synthesis of control policies for such systems. In this paper, we define two quantitative semantics for this logic, and two corresponding monitoring algorithms, which allow for real-time quantification of satisfaction of formulas by trajectories of discrete-time systems. We demonstrate the new semantics and their runtime monitors on numerical examples. 
\end{abstract}

\section{Introduction}

Temporal logics (TLs) \cite{baier2008principles} have been widely used to formulate high-level, expressive specifications for cyber-physical systems. Formal verification and synthesis algorithms have been employed to analyze and control such systems from TL specifications. In particular, Linear Temporal Logic (LTL) \cite{pnueli1977LTL} has been employed
to specify tasks for planning problems \cite{karaman16_challengesLTL,hadas2018synthesis_review,Yannis_abstractionFree,kantaros2020stylus} and for formal synthesis problems for discrete-time systems \cite{belta2017formal}. LTL formulas can be translated to automata, which can encode the progress towards task satisfaction. Automata-theoretic tools are typically used with finite abstractions of the system to produce policies that guarantee the satisfaction of tasks, or prove that they cannot be satisfied \cite{karaman16_challengesLTL,belta2017formal,Belta17_TLTL,Hadas21EventBSTL1}. Other approaches overcome some scalability issues by sampling-based planning algorithms guided by the specifications automaton, see \cite{Yannis_abstractionFree,kantaros2020stylus} where the authors use RRT$^\ast$ \cite{KaramanRRTstarIJRR} as the planning algorithm, and in \cite{Jana_scLTLRRG2021,Cristi2020RectiveRRG} the authors use RRG \cite{KaramanRRTstarIJRR} as the planning algorithm.

Signal Temporal Logic (STL)\cite{maler2004STLpaper}, Metric Temporal Logic (MTL)\cite{MTL1990}, and Time Window Temporal logic (TWTL)\cite{Cristi2017TWTL}, unlike LTL, can  express specifications with explicit, concrete-time constraints, e.g., \textit{Perform task $A$ between times $t_1$ and $t_2$ ; right after that, spend $t_5$ time units
between times $t_3$ and $t_4$ performing task $B$; and for all times do not perform task $C$}.  

The semantics of both STL and MTL are defined over real-time signals. They both have quantitative semantics, or robustness, which quantifies the degree of satisfaction of a formula by a signal \cite{donze2010robustSTL,fainekos2009MTLrobustness}. 
Most existing works that use STL and MTL for specifications find controllers by maximizing robustness, yielding runs of the system that robustly satisfy the specifications 
\cite{Belta16_QLwithSTL,Sadra_RobustSTL_MPC,fainekos2009MTLrobustness,CristiKaraman17_STL_RRTstar,Pavone2020STLCG}. The work in \cite{jana2020sampling} considers planning for syntactically co-safe LTL using RRT$^\ast$, in addition to the task specifications, other spatial requirements are expressed using fragment-STL where its robustness is used as the optimality criterion for RRT$^\ast$. In \cite{Lars21STLplnngMaxTrho_beforeAsyncTempRho}, the authors synthesize controllers for time-critical systems for which they quantify a temporal robustness measure that needs to be optimized. The traditional robustness metric is not differentiable and it is mostly determined by one value of the signal, i.e., it ``masks" most of the signal. These issues are addressed by the authors of \cite{cristi2019AGMstl}, who  introduced an arithmetic and geometric mean (AGM) robustness measure for STL.

TWTL has several advantages over STL, MTL, and other concrete-time TLs. First, its syntax and semantics can express serial tasks in an efficient and explicit way. This is important in many applications, especially in robotics  \cite{Cristi2017TWTL}. Second, TWTL formulae can be efficiently translated into automata. The complexity of the translation algorithm is independent of the formula time bounds \cite{Cristi2017TWTL}. This makes this logic suitable for automata-based synthesis and planning problems (see \cite{Cristi2020_TWTLrrt} for a planning application). TWTL, however, lacks quantitative semantics that measures the degree of satisfaction or violation of a formula. In this work, we modify the syntax of TWTL and allow it to be defined over predicated regions of the system output space. We define robustness measures to quantify the degree of satisfaction of TWTL formulae, and inspired by \cite{cristi2019AGMstl}, we extend the robustness definition to one in which we utilize the notion of AGM robustness. This enables planning and synthesis problems, which we plan to address in future follow-up work. 




Our contributions are summarized as follows. First, we adapt the ``traditional" 
quantitative semantics of STL to define a notion of sound robustness metric (Sec. \ref{subsec:TWTL_rho}). Second, inspired by the AGM-STL robustness \cite{cristi2019AGMstl}, we introduce an AGM robustness measure for TWTL, which is amenable for a wide spectrum of applications
(Sec. \ref{subsec:TWTL_eta}). Third, given partial runs of the system, i.e. runs with lengths less than the time horizon of a TWTL formula (see Definition \ref{def:phi_twtl_horizon}), we tailor the STL robustness interval semantics \cite{Seshia17_STLmtrng_robustSatInterv} to monitor the TWTL robustness (Sec. \ref{subsec:twtl_monitor}). Fourth, we introduce a similar interval sematics to monitor the AGM robustness at runtime. Finally, we validate the proposed robustness measures and their monitors in numerical examples (Sec. \ref{sec:example}).

\section{Preliminaries} 
\label{sec:prelims}

\subsection{Dynamical System}
\label{subsec:Dynmaics}

Consider a discrete-time nonlinear system in the form 
\begin{equation}\label{eq:aff_system}
	\begin{aligned}
		x_{t+\Delta t} &= f(x_t),\;t=t_0,t_0+\Delta t,t_0+2\Delta t\dots\\
            o_t &= l(x_t),
	\end{aligned}
\end{equation}
where  $x\in X\subset\mathbb{R}^d$ is the state taking values in a set $X$, $\mathbb{R}^d$ is the $d$-dimensional Euclidean space, $\Delta t\in\mathbb{R}_{>0}$, and $f:\mathbb{R}^d\rightarrow\mathbb{R}^d$. $o_t$ is an observable output of the system at time $t$, and $l(.):X\to 2^{\Pi}$ is a labeling function where $\Pi$ is a set of atomic propositions (tasks) and $2^{\Pi}$ is its power set. 

 A state trajectory of system (\ref{eq:aff_system}) is a sequence of states $\mathbf{x}:=x_{t_0}x_{t_0+\Delta t}x_{t_0+2\Delta t}\dots$ that satisfy its dynamics. In this work, an atomic proposition $\pi_A\in\Pi$ takes the Boolean value $\top$ at state $x\in X$ if $o:=l(x)\in A$ where $A:=\{o|h(o)>\sigma\}$, $h:2^\Pi\to\mathbb{R}^d$, $\sigma\in \mathbb{R}^d$, and $\bot$ otherwise. $\mathbf{x}$ generates a word, $\mathbf{o}=o_{t_0}o_{t_0+\Delta t}o_{t_0+2\Delta t}\dots\dots$. For $t_1,t_2\in\mathbb{R}_{\geq0};\;t_2:=t_1 + n_t \Delta t;\;$ where $n_t\in\mathbb{N}_{\geq0}$, we denote the corresponding system trajectory and the generated word, respectively, as the following, $\mathbf{x}_{t_1,t_2}:=x_{t_1}x_{t_1+\Delta t}\dots x_{t_2}$ and $\mathbf{o}_{t_1,t_2}:=o_{t_1}o_{t_1+\Delta t}\dots o_{t_2}$. For $t\in[t_1,t_2]$, $\mathbf{x}_{t_1,t_2}(t)=x_t$ and $\mathbf{o}_{t_1,t_2}(t):=o_t$.

\subsection{Time Window Temporal Logic}\label{subsec:TWTL}

We modify the TWTL syntax in \cite{Cristi2017TWTL} such that the atomic propositions are defined over predicated regions. The TWTL syntax is defined inductively  as follows:

\begin{equation}\label{eq:TWTL_syntax}
 \phi\;::=\;H^{d}\pi_{A}|H^{d}\neg \pi_A|\phi_1\wedge\phi_2|\phi_1\vee\phi_2|\neg\phi|\phi_1\cdot\phi_2|[\phi]^{[a,b]}   
\end{equation} where $\pi_{A}\in\Pi$ is an atomic proposition defined over the predicated region $A$; $\neg$, $\wedge$, and $\vee$ are the negation, conjunction, and disjunction Boolean operators, respectively; $H^d$ is the \textit{hold} operator; $\cdot$ is the \textit{concatenation} operator; and $[]^{[a,b]}$ is the \textit{within} operator, where $d,a,b\in\mathbb{Z}_{\geq0}$ and $a\geq b$.

The Boolean semantics over a word $\mathbf{o}_{t_1,t_2}$ is defined recursively as follows:

\begin{flalign}
\begin{aligned}\label{eq:twtl_boolean_sat}
    \mathbf{o}_{t_1,t_2}&\models H^{d}\pi_{A} \Leftrightarrow o_t\in A,\;\forall t\in[t_1,t_2]\wedge (t_2-t_1\geq d\Delta t) \\
    \mathbf{o}_{t_1,t_2}&\models H^{d}\neg\pi_{A} \Leftrightarrow o_t\notin A,\;\forall t\in[t_1,t_2]\wedge (t_2-t_1\geq d\Delta t)  \\
    \mathbf{o}_{t_1,t_2}&\models\phi_1\wedge\phi_2\Leftrightarrow(\mathbf{o}_{t_1,t_2}\models\phi_1)\wedge(\mathbf{o}_{t_1,t_2}\models\phi_2)   \\
     \mathbf{o}_{t_1,t_2}&\models\phi_1\vee\phi_2\Leftrightarrow(\mathbf{o}_{t_1,t_2}\models\phi_1)\vee(\mathbf{o}_{t_1,t_2}\models\phi_2)   \\
     \mathbf{o}_{t_1,t_2}&\models\neg\phi\Leftrightarrow\mathbf{o}_{t_1,t_2}\not\models\phi       \\
     \mathbf{o}_{t_1,t_2}&\models\phi_1\cdot\phi_2\Leftrightarrow\exists t=\argmin_{t\in[t_1,t_2]}\{\mathbf{o}_{t_1,t}\models\phi_1\}\\&\wedge(\mathbf{o}_{t+\Delta t,t_2}\models\phi_2)\\
     \mathbf{o}_{t_1,t_2}&\models[\phi]^{[a,b]}\Leftrightarrow\exists t\geq t_1+a \text{ s.t. } \mathbf{o}_{t,t_1+b}\models\phi\\&\wedge(t_2-t_1\geq b)
\end{aligned}
\end{flalign}
\begin{definition} [TWTL Time Horizon \cite{Cristi2017TWTL}]\label{def:phi_twtl_horizon}
Given $\phi$, the time horizon is defined recursively as follows. 
\begin{equation}
    ||\phi||:=\begin{cases}
        \max(||\phi_1||,||\phi_2||); & \text{if }\phi\in\{\phi_1\wedge\phi_2,\phi_1\vee\phi_2\}\\
        ||\phi_1||;& \text{if } \phi = \neg \phi_1\\
        ||\phi_1||+||\phi_2||+\Delta t;&\text{if } \phi = \phi_1\cdot\phi_2\\
        d\Delta t;& \text{if }\phi=H^d\pi_{A}\\
        b;&\text{if }\phi = [\phi_1]^{[a,b]}
    \end{cases}
\end{equation}
\end{definition}
 \section{TWTL Quantitative Semantics}\label{sec:twtl_rho_eta}
 
Inspired by STL robustness \cite{donze2010robustSTL} and its AGM version \cite{cristi2019AGMstl}, we tailor robustness measures for TWTL formulae to reason about the degree of satisfaction of TWTL tasks.
\subsection{TWTL Robustness}\label{subsec:TWTL_rho}
\begin{definition}(TWTL Robustness)\label{def:traditional_robustness}
Given a TWTL formula $\phi$ and an output word $\mathbf{o}_{t_1,t_2}$ of system (\ref{eq:aff_system}), we define the robustness degree $\rho(\mathbf{o}_{t_1,t_2},\phi)$ at time $0$, recursively, as follows: 
\begin{flalign}
\begin{aligned}\label{eq:TWTL_robustness}
    & \rho(\mathbf{o}_{t_1,t_2},H^{d}\pi_{A}) :=\begin{cases}
    \min\limits_{t\in[t_1,d+t_1]}h(o_t) & ; (t_2-t_1\geq d\Delta t)         \\
    \rho_{\bot} & ;\text{otherwise}
    \end{cases}\\
    &\rho(\mathbf{o}_{t_1,t_2},\phi_1\wedge\phi_2):=\min\{ \rho(\mathbf{o}_{t_1,t_2},\phi_1),\rho(\mathbf{o}_{t_1,t_2},\phi_2)\}          \\
     &\rho(\mathbf{o}_{t_1,t_2},\phi_1\vee\phi_2):=\max\{ \rho(\mathbf{o}_{t_1,t_2},\phi_1),\rho(\mathbf{o}_{t_1,t_2},\phi_2)\}          \\
     &\rho(\mathbf{o}_{t_1,t_2},\neg \phi)=-\rho(\mathbf{o}_{t_1,t_2},\phi)\\
     & \rho(\mathbf{o}_{t_1,t_2},\phi_1\cdot\phi_2) := \\& \quad \max_{t\in [t_1, t_2)} \left\{\min\{\rho(\mathbf{o}_{t_1,t},\phi_1), \rho(\mathbf{o}_{t+1,t_2},\phi_2) \} \right\}\\
    &\rho(\mathbf{o}_{t_1,t_2},[\phi]^{[a,b]}) :=\begin{cases}
    \max\limits_{t\geq t_1+a}\{\rho(\mathbf{o}_{t,t_1+b},\phi)\};(t_2-t_1\geq b)
    \\
    \rho_{\bot};\; \text{otherwise}
    \end{cases}
\end{aligned}
\end{flalign}
where $\rho_{\bot}$ denotes a large negative value that indicates the robustness of Boolean $\bot$.     
\end{definition}
For $\mathbf{o}_{t_1,t_2}$ and $\phi$, the robustness value $\rho(\mathbf{o}_{t+1,t_2},\phi)$ indicates how far is $\mathbf{o}_{t_1,t_2}$ from the decision boundary of the predicated region of the task. A positive $\rho(\mathbf{o}_{t+1,t_2},\phi)$ implies the Boolean satisfaction of the task, where the greater the value the more robustly $\rho(\mathbf{o}_{t+1,t_2},\phi)$ satisfies the task. A similar argument can be made for negative robustness for violation of a task.   
\begin{lemma}\label{lema:soundness_rho}
TWTL robustness (\ref{eq:TWTL_robustness}) is sound, i.e., the Boolean satisfaction (violation) is implied by a positive (negative) robustness value.  
\end{lemma}
\begin{proof}
We prove the soundness by structural induction over the formula $\phi$. The base case for TWTL is the \textit{Hold} case, which we prove as follows. 
    
    \noindent \textit{Hold:} First, if $t_2-t_1<d\Delta t$ then $\rho(\mathbf{o}_{t_1,t_2},H^d\pi_{A})=\rho_\bot<0$ which implies that $\mathbf{o}_{t_1,t_2}\not\models\phi$. Second, if $t_2-t_1\geq d\Delta t$, let $\rho(\mathbf{o}_{t_1,t_2},H^d\pi_{A})>0$. Assume that there is $t^\prime\in[t_1,t_1+d\Delta t]$ such that $\rho(\mathbf{o}_{t^\prime},\pi_{A})<0$ then we get $\rho(\mathbf{o}_{t_1,t_2},H^d\pi_{A})=\min\limits_{t\in[t_1,t_1+d\Delta t]}h(o_t)<0$, which is a contradiction, then $\mathbf{o}_{t_1,t_2}\models H^d\pi_A$. Third, if $t_2-t_1\geq d\Delta t$, let $\rho(\mathbf{o}_{t_1,t_2},H^d\pi_{A})<0$. Assume that for all $t^\prime\in[t_1,t_1+d\Delta t]$ such that $\rho(\mathbf{o}_{t^\prime},\pi_{A})>0$ then $\rho(\mathbf{o}_{t_1,t_2},H^d\pi_{A})=\min\limits_{t\in[t_1,t_1+\Delta t]}h(o_t)>0$, which is a contradiction, then $\mathbf{o}_{t_1,t_2}\not\models H^d\pi_A$.  
    
    Then we have the following induction cases: 
    
    \noindent \textit{Conjunction:} Let $\rho(\mathbf{o}_{t_1,t_2},\phi_1\wedge\phi_2)>0$. Assume that one or both $\rho(\mathbf{o}_{t_1,t_2},\phi_i)<0$, $i=1,2$, then $\rho(\mathbf{o}_{t_1,t_2},\phi_1\wedge\phi_2)=\min\{ \rho(\mathbf{o}_{t_1,t_2},\phi_1),\rho(\mathbf{o}_{t_1,t_2},\phi_2)\}<0$, which is a contradiction. Then by induction $\mathbf{o}_{t_1,t_2}\models\phi_1\wedge\phi_2$. In the case when $\rho(\mathbf{o}_{t_1,t_2},\phi_1\wedge\phi_2)<0$. Assume that $\rho(\mathbf{o}_{t_1,t_2},\phi_i)>0$, $i=1,2$, then $\rho(\mathbf{o}_{t_1,t_2},\phi_1\wedge\phi_2)=\min\{ \rho(\mathbf{o}_{t_1,t_2},\phi_1),\rho(\mathbf{o}_{t_1,t_2},\phi_2)\}>0$, which is a contradiction. Then by induction $\mathbf{o}_{t_1,t_2}\not\models\phi_1\wedge\phi_2$.   
    
    \noindent \textit{Disjunction:} Follows similarly to the \textit{conjunction} case. 
    
    \noindent \textit{Negation:} Let $\phi=\neg\varphi$ and $\rho(\mathbf{o}_{t_1,t_2},\phi)>0$. We get  $\rho(\mathbf{o}_{t_1,t_2},\varphi)<0$, and by induction hypothesis $\mathbf{o}_{t_1,t_2}\not\models\varphi$; hence, $\mathbf{o}_{t_1,t_2}\models\phi$. Similarly, if $\rho(\mathbf{o}_{t_1,t_2},\phi)>0$ we get $\mathbf{o}_{t_1,t_2}\models\varphi$; hence, $\mathbf{o}_{t_1,t_2}\not\models\phi$.    
    
    \noindent \textit{Within:} Similar to the \textit{hold} case. 
    
    \noindent \textit{Concatenation:} Given that the language of TWTL is assumed to be unambiguous \cite{Cristi2017TWTL}, we modify the syntax of the concatenation operator as follows. 
    \begin{equation}\label{eq:disjConjunCat}
    \mathbf{o}_{t_1,t_2} \models\phi_1\cdot\phi_2\Leftrightarrow\bigvee\limits_{t\in[t_1,t_2]}\mathbf{o}_{t_1,t}\models\phi_1\wedge\mathbf{o}_{t+\Delta t,t_2}\models\phi_2
    \end{equation}
    Thus, the soundness of $\rho(\mathbf{o}_{t_1,t_2},\phi_1\cdot\phi_2)$ follows directly from the soundness of the robustness of the \textit{disjunction} and \textit{conjunction} cases. 
\end{proof}
\begin{example}\label{eg:rho}
Consider a TWTL formula $\phi= [H^6 \pi_A]^{[0,10]}$, where $A=\{o \mid o\geq4\}$, which reads as: ``\textit{Within} time $0$ and time $10$, \textit{hold} in $\pi_A$ for $6$ time steps"; and three output words $\mathbf{o}_1$, $\mathbf{o}_2$ and $\mathbf{o}_3$, which are depicted as blue, green, and red traces in the top-left figure of Fig. \ref{fig:allmntringrhoResults}, respectively. One can see that $\mathbf{o}_1$ and $\mathbf{o}_2$ satisfy the task specification where $\rho(\mathbf{o}_1,\phi)=\rho(\mathbf{o}_2,\phi)=0.099$, whereas $\mathbf{o}_2$ violates the task, where $\rho(\mathbf{o}_2,\phi)=-2$.   
\end{example}
\subsection{TWTL Arithmetic and Geometric Mean Robustness }
\label{subsec:TWTL_eta}
TWTL robustness (Definition \ref{def:traditional_robustness}) accounts for the most critical points of the system output word, which is necessary for the soundness of the robustness (see Lemma \ref{lema:soundness_rho}).
However, it may be very pessimistic robustness measure as highlighted in Example~\ref{eg:eta}.
for instance, the computation of the robustness $\rho(\mathbf{o}_{t_1,t_2},H^d \pi_A)$ is dominated by the minimum valuation of the predicate function $h(\cdot)$ over the system word $\mathbf{o}$.
Moreover, since its computation involves $\mathrm{min}$ and $\mathrm{max}$, it leads to a non-smooth measure
which is computationally challenging for heuristic- and gradient-based approaches to maximizes the robustness of the overall task.  

To this end, we tailor the notion of AGM robustness \cite{cristi2019AGMstl} to define the AGM robustness measure for TWTL $\eta$. As we show in the following, $\eta$ helps mitigate some of the shortcomings of TWTL robustness $\rho$
and provides a more optimistic, smooth and sound robustness measure for TWTL.

Consider the function $F:\mathbb{R}\to\mathbb{R}$, and let $[F]_{+}:=\begin{cases}F;\;F>0\\0;\; \text{otherwise}\end{cases}$ and $[F]_\_=-[-F]_{+}$, where $F=[F]_{+}+[F]_\_$.
We define AGM functions of disjunction and conjunction of $r_i\in\mathbb R$, $i=1,\dots,N$, respectively, as follows.  
\begin{flalign}
\begin{aligned}\label{eq:AGM_dis}
        &\mathrm{AGM}_{\vee}(r_1,\dots,r_N):= \begin{cases}
                -\sqrt[N]{\prod\limits_{i=1}^N(1-r_i)}+1;\\ \ \ \text{if } \forall i\in\{1,\ldots,N\}, r_i < 0\\\frac{1}{N}\sum\limits_{i=1}^N [r_i]_+;\ \ \text{otherwise }
        \end{cases}\\
    \end{aligned}
\end{flalign}
\begin{flalign}
\begin{aligned}\label{eq:AGM_con}
        &\mathrm{AGM}_{\wedge}(r_1,\dots,r_N):= \begin{cases}
            \sqrt[N]{\prod\limits_{i=1}^N(1+r_i)}-1;\\ \ \ \text{if } \forall i\in\{1,\ldots,N\}, r_i > 0\\
                \frac{1}{N}\sum\limits_{i=1}^N [r_i]_\_;\ \ \text{otherwise}
        \end{cases}\\
    \end{aligned}
\end{flalign}
\begin{definition}[Normalized TWTL formulae]
    Given syntax~\eqref{eq:TWTL_syntax}, a normalized TWTL formula $\phi$ is presented as the formulae defined over normalized atomic propositions $\pi_{A_{\mathrm{norm}}}$, where
    $A_{\mathrm{norm}}:=\{o \mid h_{\mathrm{norm}}(o)>\sigma_n\}$, $h_{\mathrm{norm}}:2^\Pi\to[-1,1]^{d}$, and $\sigma_{\mathrm{norm}}\in [-1,1]$. 
\end{definition}

Throughout the rest of the paper, we assume that all TWTL formulae are normalized unless explicitly stated otherwise.

\begin{definition}(TWTL Arithmetic-Geometric Mean Robustness)\label{def:AGM_robustness}
Given a normalized TWTL formula $\phi$, we define the AGM robustness of the output word $\mathbf{o}_{t_1,t_2}$ with respect to $\phi$, recursively, using~\eqref{eq:AGM_dis} and~\eqref{eq:AGM_con}.
\end{definition}
\begin{flalign}
\begin{aligned}\label{eq:TWTL_AGMrobustness1}
     &\eta(\mathbf{o}_{t_1,t_2},\top):=+1\\
    &\eta(\mathbf{o}_{t_1,t_2},\bot):=-1\\
    &\eta(\mathbf{o}_{t},\pi_{A}):=\frac{1}{2}(h(\mathbf{o}_t)-\sigma_n)\\&\eta(\mathbf{o}_{t_1,t_2},\phi_1\wedge\phi_2):=\mathrm{AGM}_{\wedge}(\eta(\mathbf{o}_{t_1,t_2},\phi_1),\eta(\mathbf{o}_{t_1,t_2},\phi_2))\\
    &\eta(\mathbf{o}_{t_1,t_2},\phi_1\vee\phi_1):=\mathrm{AGM}_{\vee}(\eta(\mathbf{o}_{t_1,t_2},\phi_1),\eta(\mathbf{o}_{t_1,t_2},\phi_2))\\
    &\eta(\mathbf{o}_{t_1,t_2},\neg\phi):=-\eta(\mathbf{o}_{t_1,t_2},\phi)
\end{aligned}
\end{flalign}
\begin{strip}
\begin{flalign}
\begin{aligned}\label{eq:TWTL_AGMrobustness2}
    &\eta(\mathbf{o}_{t_1,t_2},H^d\pi_{A}):=\begin{cases}
        \mathrm{AGM}_{\wedge}(\eta(\mathbf{o}_{t},\pi_{A})|t\in[t_1,t_1+d\Delta t]);\text{ if }(t_2-t_1)\geq d\\
        -1;\text{ otherwise}
    \end{cases}
    \\& \eta(\mathbf{o}_{t_1,t_2},[\phi]^{[a,b]}]):=\begin{cases}
        \mathrm{AGM}_{\vee}(\eta(\mathbf{o}_{t,t_1+b},\phi)|t\in[t_1+a,t_1+b]);\text{ if } (t_2-t_1)\geq b\\ 
        -1; \text{ otherwise}
    \end{cases} 
    \\& \eta(\mathbf{o}_{t_1,t_2},\phi_1.\phi_2):= \mathrm{AGM}_{\vee}(\mathrm{AGM}_{\wedge}(\eta(\mathbf{o}_{t_1,t},\phi_1),\eta(\mathbf{o}_{t+\Delta t,t_2},\phi_2))|t\in [t_1,t_2))
\end{aligned}
\end{flalign}
\end{strip}

\begin{theorem}\label{thm:AGMrobustness_soundness}
TWTL AGM robustness, Definition \ref{def:AGM_robustness}, is sound.
Formally, we have
\begin{equation*}
\begin{aligned}
    \eta(\mathbf{o}_{t_1,t_2},\phi)>0\Leftrightarrow\rho(\mathbf{o}_{t_1,t_2},\phi)>0 \implies\mathbf{o}_{t_1,t_2}\models\phi\\
    \eta(\mathbf{o}_{t_1,t_2},\phi)<0\Leftrightarrow\rho(\mathbf{o}_{t_1,t_2},\phi)<0 \implies\mathbf{o}_{t_1,t_2}\not\models\phi
\end{aligned}  
\end{equation*}
\end{theorem}
\begin{proof}
We prove the soundness by structural induction over the formula $\phi$. The base case for TWTL is the \textit{Hold} case, which we prove as follows.

\noindent \textit{Hold:} First, if $t_2-t_1<d\Delta t$ then $\eta(\mathbf{o}_{t_1,t_2},H^d\pi_{A})=-1$ which implies that $\mathbf{o}_{t_1,t_2}\not\models\phi$. Second, if $t_2-t_1\geq d\Delta t$, let $\eta(\mathbf{o}_{t_1,t_2},H^d\pi_{A})>0$. Assume that there is $t^\prime\in[t_1,t_1+d\Delta t]$ such that $\eta(\mathbf{o}_{t^\prime},\pi_{A})<0$ then by (\ref{eq:AGM_con}) we get $\eta(\mathbf{o}_{t_1,t_2},H^d\pi_{A})=\frac{1}{|[t_1,t_2]|}\sum\limits_{t^\prime\in[t_1,t_1+d\Delta t]}[\eta(\mathbf{o}_{t^\prime},\pi_{A})]_\_<0$, which is a contradiction, then $\mathbf{o}_{t_1,t_2}\models H^d\pi_A$. Third, if $t_2-t_1\geq d$, let $\eta(\mathbf{o}_{t_1,t_2},H^d\pi_{A})<0$. Assume that for all $t^\prime\in[t_1,t_1+d\Delta t]$ such that $\eta(\mathbf{o}_{t^\prime},\pi_{A})>0$ then by (\ref{eq:AGM_dis}) we get $\eta(\mathbf{o}_{t_1,t_2},H^d\pi_{A})={}^{|[t_1,t_1+d\Delta t]|}\sqrt{\prod\limits_{t^\prime\in[t_1,t_1+d\Delta t]}(1+\eta(\mathbf{o}_t,\pi_A)}-1>0$, which is a contradiction, then $\mathbf{o}_{t_1,t_2}\not\models H^d\pi_A$.  

    Then we have the following induction cases: 
    
    \noindent The soundness for the cases $\top$, $\bot$ and $\pi_A$, and the \textit{conjunction}, \textit{disjunction}, and \textit{negation} cases, follow directly from Theorem 2 in \cite{cristi2019AGMstl}. 
    
    \noindent \textit{Within:} Follows similarly to the \textit{hold} case. 
    
    \noindent \textit{Concatenation:} Given the Boolean semantic (\ref{eq:disjConjunCat}), the soundness of $\eta(\mathbf{o}_{t_1,t_2},\phi_1\cdot\phi_2)$ follows directly from the soundness of the AGM robustness of the \textit{disjunction} and \textit{conjunction} cases.
\end{proof}

\begin{example}(Continued)\label{eg:eta}
    Consider the same formula and words of Example \ref{eg:rho}. $\eta(\mathbf{o}_1,\phi)=0.061$, $\eta(\mathbf{o}_2,\phi)=0.010$, where, unlike their $\rho$ value, the AGM robustness measure $\eta$ rewards words with more satisfying valuations instead of being dominated by the most critical valuations while also preserving the soundness property. The word $\mathbf{o}_1$ (the blue trace in top-left figure of Fig. \ref{fig:allmntringrhoResults}) has more valuations that robustly contribute to satisfying  the formula.
    In $\rho(\mathbf{o}_1,\phi)$, the  $4^{\mathrm{th}}$ point of the trace, which is close to lead to violating the task, dominates the robustness computation. Even if we assume that the $4^{\mathrm{th}}$ is $2$ (which would lead to violating the task), its $\eta$ value would be $-0.35$ that is higher than its corresponding $\rho$ value, $-2$. 
    The computation of $\eta$ considers the fact that the trace has promising valuations which contribute to satisfying the task.
    For $\mathbf{o}_3$ (the red trace in the same figure), on the other hand, $\eta(\mathbf{o}_3,\phi)=-0.61$ which is higher than $\rho(\mathbf{o}_3,\phi)=-2$, is more realistic violation measure given that some valuations of $\mathbf{o}_3$ are close to contributing in satisfying the task.        
\end{example}

\section{Runtime Robustness Monitoring}\label{subsec:twtl_monitor}

 Considering runs of the system with time horizon less than TL specifications time horizon, runtime verification techniques are introduced as light weight algorithms to monitor the satisfaction given such partial runs, see \cite{jyo18SurveyMntring} for review on monitoring different TLs. Different techniques are usually utilized for the monitoring task, the work in \cite{b2022runtimeMonitForTWTL} uses a rewriting technique to monitor the Boolean satisfaction of TWTL. In \cite{Seshia17_STLmtrng_robustSatInterv}, the authors introduced interval semantics to monitor the robustness degree of STL specifications. The technique considers partial runs, and with the set of all possible completions  of the run, it computes the best and worst possible robustness.

 In this work, we tailor the STL robustness monitor from \cite{Seshia17_STLmtrng_robustSatInterv} to encode a robustness interval $[\rho]$ to monitor TWTL robustness.
 Similarly, we introduce an AGM robustness interval $[\eta]$ to monitor the AGM robustness. Our monitors are sound, which means the correct (AGM) robustness belongs to the produced interval of the runtime monitor at any time step.   
 
Let us consider some preliminary definitions that we use in our TWTL interval semantics. 
\begin{definition}[Prefix, Completions]
    Consider the time horizon $||\phi||$ and output words $\mathbf{o}_{t_1,t^\prime}$ and $\mathbf{o}_{t_1,t_2}$, where $t^\prime<||\phi||$ and $t_2\geq||\phi||$. We denote $\mathbf{o}_{t_1,t^\prime}$ as a prefix of $\mathbf{o}_{t_1,t_2}$ if $\forall t\in[t_1,t^\prime],\mathbf{o}_{t_1,t_2}(t)=\mathbf{o}_{t_1,t^\prime}(t)$; consequently we define a set of all possible completions of a prefix as $\mathfrak{C}:=\{\mathbf{o}_{t_1,t_2} \mid \mathbf{o}_{t_1,t^\prime} \text{ is a prefix of } \mathbf{o}_{t_1,t_2}\}$ 
\end{definition}

\begin{definition} [Arithmetics on interval semantics] 
    Consider the following set of intervals $\mathbf{I}:=\{I_i\}^N_{i=1}$, where $I_i:=[\underline{I}_i,\bar{I}_i]$ and $\underline{I}_i\leq\bar{I}_i$. We define following arithmetics over $\mathbf{I}$
    \begin{equation}
        \begin{aligned}\label{eq:AGM_interval_Arithm}
            \mathrm{\mathbf{AGM}}_{\vee}(\mathbf{I}):= [\mathrm{AGM}_{\vee}(\underline{I}_1,\dots,\underline{I}_N),\mathrm{AGM}_{\vee}(\bar{I}_1,\dots,\bar{I}_N)],\\
            \mathrm{\mathbf{AGM}}_{\wedge}(\mathbf{I}):= [\mathrm{AGM}_{\wedge}(\underline{I}_1,\dots,\underline{I}_N),\mathrm{AGM}_{\wedge}(\bar{I}_1,\dots,\bar{I}_N)].
        \end{aligned}
    \end{equation}
    \begin{equation}
        \begin{aligned}\label{eq:minmax_interval_Arithm}
             \mathrm{\mathbf{max}}(\mathbf{I}):= [ \max(\underline{I}_1,\dots,\underline{I}_N),\max(\bar{I}_1,\dots,\bar{I}_N)],\\
             \mathrm{\mathbf{min}}(\mathbf{I}):= [ \min(\underline{I}_1,\dots,\underline{I}_N), \min(\bar{I}_1,\dots,\bar{I}_N)].
        \end{aligned}
    \end{equation}
\end{definition}

\smallskip
The singleton interval $[I, I]$ is denoted by $\{I\}$.

    
Before introducing the recursive definition of $[\rho]$ and $[\eta]$ we introduce the following definition. 
\begin{definition}(Compact Representation of Set of Intervals)
    Given words $\mathbf{o}_{t_1,t^\prime}$, $\mathbf{o}_{t,t^\prime}$, $t\in[t_a,t_b]$, and formulae $\phi$ and $\phi_i$, $i=1,\dots,N$, for $[\rho]$ and $[\eta]$ we denote the following set representation:  
     \begin{equation}
        \begin{aligned}
            [.]^{i=1,\dots,N}(\mathbf{o}_{t_1,t^\prime},\phi_i)=\{[.](\mathbf{o}_{t_1,t^\prime},\phi_1),\dots,[.](\mathbf{o}_{t_1,t^\prime},\phi_N)\},\\
            [.]_{t_a:t_b}(\mathbf{o}_{t,t^\prime},\phi)=\{[.](\mathbf{o}_{t_a,t^\prime},\phi_1),\dots,[.](\mathbf{o}_{t_b,t^\prime},\phi)\},
        \end{aligned}
    \end{equation}
where $[.](\mathbf{o_{t_1,t_2}},\phi)$ is the interval semantics defined next. 
\end{definition}

Consider a TWTL formula $\phi$ with time horizon $||\phi||$. 
Given a partial word $\mathbf{o}_{t_1,t^\prime}$, where $t^\prime\leq||\phi||$, for monitoring the robustness $\rho$ at time $t^\prime$, we define the robustness interval $[\rho]$ recursively as follows. 

\begin{flalign}
\begin{aligned}\label{eq:rho_mntring}
    & [\rho](\mathbf{o}_{t_1,t_2},H^{d}\pi_{A}) :=\begin{cases}
    \{\rho(\mathbf{o}_{t_1,t_2},H^{d}\pi_{A})\}; \ (t_2-t_1\geq d)\\
    [\rho_\bot,\min\limits_{t\in[t_1,d+t_1]}h(o_t)]; \ \text{otherwise}
    \end{cases}\\
    &[\rho](\mathbf{o}_{t_1,t_2},\phi_1\wedge\phi_2):=\mathrm{\mathbf{min}}( [\rho](\mathbf{o}_{t_1,t_2},\phi_1),[\rho](\mathbf{o}_{t_1,t_2},\phi_2))\\
     &[\rho](\mathbf{o}_{t_1,t_2},\phi_1\vee\phi_2):=\mathrm{\mathbf{max}}([\rho](\mathbf{o}_{t_1,t_2},\phi_1),[\rho](\mathbf{o}_{t_1,t_2},\phi_2))          \\
     & [\rho](\mathbf{o}_{t_1,t_2},\phi_1\cdot\phi_2) :=\\&  \mathrm{\mathbf{max}}_{t\in [t_1, t_2)} ( \mathrm{\mathbf{min}}([\rho](\mathbf{o}_{t_1,t},\phi_1),[\rho](\mathbf{o}_{t+1,t_2},\phi_2) ) )\\
    &[\rho](\mathbf{o}_{t_1,t_2},[\phi]^{[a,b]}) :=\begin{cases}
    \{\rho(\mathbf{o}_{t_1,t_2},[\phi]^{[a,b]}\}; \ (t_2-t_1\geq b)\\
    [\max\limits_{t\geq t_1+a}\{\rho(\mathbf{o}_{t,t_1+b},\phi)\},\rho_\top];\;\\ \qquad \text{otherwise}&
    \end{cases}
\end{aligned}
\end{flalign}

Considering the same specification and word, for monitoring the AGM robustness $\eta$ at time $t^\prime$, we define the robustness interval $[\eta]$, recursively, using (\ref{eq:AGM_runtime_mntr1}),(\ref{eq:AGM_runtime_mntr2}). 
\begin{flalign}
\begin{aligned}\label{eq:AGM_runtime_mntr1}
    &[\eta](\mathbf{o}_{t_1,t^\prime},\phi_1\wedge\phi_2,t^\prime):=  \mathrm{\mathbf{AGM}}_{\wedge}([\eta](\mathbf{o}_{t_1,t^\prime},\phi_i)^{i=1,2})\\
    &[\eta](\mathbf{o}_{t_1,t^\prime},\phi_1\vee\phi_2,t^\prime):=  \mathrm{\mathbf{AGM}}_{\vee}([\eta](\mathbf{o}_{t_1,t^\prime},\phi_i)^{i=1,2})\\
     &[\eta](\mathbf{o}_{t_1,t^\prime},H^d \pi_{A},t^\prime):=[\underline{\eta},\bar{\eta}]
\end{aligned}
\end{flalign}
\begin{strip}
\begin{flalign}
\begin{aligned}\label{eq:AGM_runtime_mntr2}
    &\bar{\eta}:= \begin{cases}
    {}^{d+1}\sqrt{\prod\limits_{t\in[t_1,t^\prime]}(1+\eta(\mathbf{o}_{t},\pi_A))(1+\eta_{\mathrm{max}})^{|[t^\prime,d]|}}-1 ; &\text{if } \forall t\in [t_1,t^\prime], \eta(\mathbf{o}_{t},\pi_{A})>0 \wedge (t^\prime-t_1)<d \\
      \frac{1}{d+1}\sum\limits_{t\in[t_1,t^\prime]}[\eta(\mathbf{o}_t,\pi_A)]_\_; &\text{ if }\exists t\in [t_1,t^\prime], \eta(\mathbf{o}_{t},\pi_{A})<0\wedge (t^\prime-t_1)<d \\
      \eta(\mathbf{o}_{t_1,t^\prime},H^d\pi_{A});&\text{If }(t^\prime-t_1)\geq d
    \end{cases} \\
    &\underline{\eta}:=\begin{cases}
    \frac{|(t^\prime,d]|}{|[t_1,d]|}\eta_{\mathrm{min}}; &\text{if } \forall t\in [t_1,t^\prime], \eta(\mathbf{o}_{t},\pi_{A})>0 \wedge (t^\prime-t_1)<d \\
    \frac{1}{d+1}\left(\sum\limits_{t\in[t_1,t^\prime]}[\eta(\mathbf{o}_t,\pi_A)]_\_+|[t^\prime,d]|\eta_{\mathrm{min}}\right); &\text{if }\exists t\in [t_1,t^\prime], \eta(\mathbf{o}_{t},\pi_{A})<0 \wedge (t^\prime-t_1)<d\\
      \eta(\mathbf{o}_{t_1,t^\prime},H^d\pi_{A});&\text{If }(t^\prime-t_1)\geq d
    \end{cases}\\
    &[\eta](\mathbf{o}_{t_1,t^\prime},[\phi]^{[a,b]},t^\prime):=\begin{cases}
    \{\eta(\mathbf{o}_{t_1,t^\prime},[\phi]^{[a,b]})\};&\text{if }(t^\prime-t_1)\geq b\\
     \mathrm{\mathbf{AGM}}_{\vee}([\eta](\mathbf{o}_{t,t_1+b},\phi)_{t_1+a:t_2}); &\text{otherwise}\end{cases}\\
&[\eta](\mathbf{o}_{t_1,t^\prime},\phi_1.\phi_2,t^\prime):=  \mathrm{\mathbf{AGM}}_{\vee}( \mathrm{\mathbf{AGM}}_{\wedge}([\eta](\mathbf{o}_{t_1,t^\prime-\Delta t},\phi_1),[\eta](\mathbf{o}_{t^\prime,t^\prime},\phi_2)),\dots, \mathrm{\mathbf{AGM}}_{\wedge}([\eta](\mathbf{o}_{t_1,t_1},\phi_1),[\eta](\mathbf{o}_{t_1+\Delta t,t^\prime},\phi_2)))
\end{aligned}
\end{flalign}
\end{strip}


\begin{theorem}\label{thm:soundness_of_monitors}
Consider a TWTL formula $\phi$ and a partial word $\mathbf{o}_{t_1,t^\prime}$, then for any possible completion word $\mathbf{o}_{t_1,t_2}\in\mathfrak{C}$, $\rho(\mathbf{o}_{t_1,t_2},H^d\pi_{A})\in[\rho](\mathbf{o}_{t_1,t^\prime},H^d\pi_{A})$ and $\eta(\mathbf{o}_{t_1,t_2},H^d\pi_{A})\in[\eta](\mathbf{o}_{t_1,t^\prime},H^d\pi_{A})$.
\end{theorem}
\begin{proof}
\textbf{First}, we prove that for any $\mathbf{o}_{t_1,t_2}\in\mathfrak{C}$, $\rho(\mathbf{o}_{t_1,t_2},H^d\pi_{A})\in[\rho](\mathbf{o}_{t_1,t^\prime},H^d\pi_{A})$ by structural induction over formula $\phi$. The base case for TWTL is the \textit{Hold} case, which we prove as follows.

\noindent \textit{Hold:} We prove the following cases by contradiction. Given partial word $\mathbf{o}_{t_1,t^\prime}$, then for any possible completion $\mathbf{o}_{t_1,t_2}\in\mathfrak{C}$, we assume that $\rho(\mathbf{o}_{t_1,t_2},H^d\pi_{A})\not\in[\rho](\mathbf{o}_{t_1,t^\prime},H^d\pi_{A})$.
First, if $t^\prime-t_1\geq d$, using \eqref{eq:rho_mntring},
$[\rho](\mathbf{o}_{t_1,t^\prime},H^d\pi_{A})=\{\rho(\mathbf{o}_{t_1,t_2},H^d\pi_{A}\}$ which is a contradiction.
Thus, $\rho(\mathbf{o}_{t_1,t_2},H^d\pi_{A})\in[\rho](\mathbf{o}_{t_1,t^\prime},H^d\pi_{A})$.
Second, if $t^\prime-t_1<d$, using \eqref{eq:rho_mntring}, $[\rho](\mathbf{o}_{t_1,t^\prime},H^d\pi_{A})= [\rho_\bot,\min\limits_{t\in[t_1,d+t_1]}h(o_t)]$, then by Definition \ref{def:traditional_robustness}, $\rho(\mathbf{o}_{t_1,t_2},H^d\pi_{A})>\rho_\bot$ and $\rho(\mathbf{o}_{t_1,t_2},H^d\pi_{A})\leq\min\limits_{t\in[t_1,d+t_1]}h(o_t)$ which is a contradiction.
Hence, $\rho(\mathbf{o}_{t_1,t_2},H^d\pi_{A})\in[\rho](\mathbf{o}_{t_1,t^\prime},H^d\pi_{A})$.

Then we have the following induction cases:

\noindent \textit{Within:} Follows similarly to the \textit{hold} case.

\noindent \textit{Conjunction/Disjunction:} The soundness follows by the properties of the interval arithmetic \eqref{eq:minmax_interval_Arithm}.

\noindent \textit{Concatenation:} Given the Boolean semantic~\eqref{eq:disjConjunCat}, the soundness of concatenation follows directly from the soundness of the \textit{Conjunction/Disjunction} case.

\textbf{Second,}  we prove that for any $\mathbf{o}_{t_1,t_2}\in\mathfrak{C}$, $\eta(\mathbf{o}_{t_1,t_2},H^d\pi_{A})\in[\eta](\mathbf{o}_{t_1,t^\prime},H^d\pi_{A})$ by structural induction over formula $\phi$. The base case for TWTL is the \textit{Hold} case, which we prove as follows.

\noindent \textit{Hold:} We prove the following cases by contradiction.
Given partial word $\mathbf{o}_{t_1,t^\prime}$, then for any possible completion $\mathbf{o}_{t_1,t_2}\in\mathfrak{C}$, we assume that $\eta(\mathbf{o}_{t_1,t_2},H^d\pi_{A})\not\in[\eta](\mathbf{o}_{t_1,t^\prime},H^d\pi_{A})$.
First, if $t^\prime-t_1\geq d$, using \eqref{eq:AGM_runtime_mntr2},
$[\eta](\mathbf{o}_{t_1,t^\prime},H^d\pi_{A}) = \{\eta(\mathbf{o}_{t_1,t_2},H^d\pi_{A})\}$ which is a contradiction.
Thus, $\eta(\mathbf{o}_{t_1,t_2},H^d\pi_{A}) \in [\eta](\mathbf{o}_{t_1,t^\prime},H^d\pi_{A})$.
Second, if $\forall t\in [t_1,t^\prime], \eta(\mathbf{o}_{t},\pi_{A})>0 \wedge (t^\prime-t_1)<d$, using \eqref{eq:AGM_runtime_mntr2},
{\small $[\eta](\mathbf{o}_{t_1,t^\prime},H^d\pi_{A})= \left[ \frac{|(t^\prime,d]|}{|[t_1,d]|}\eta_{\mathrm{min}},    {}^{d+1}\sqrt{\prod\limits_{t\in[t_1,t^\prime]}(1+\eta(\mathbf{o}_{t},\pi_A))(1+\eta_{\mathrm{max}})^{|[t^\prime,d]|}}-1\right]$}.
Then by Definition \ref{def:AGM_robustness}, $\eta(\mathbf{o}_{t_1,t_2},H^d\pi_{A})\in[\eta](\mathbf{o}_{t_1,t^\prime},H^d\pi_{A})$ which is a contradiction. 
Hence, $\eta(\mathbf{o}_{t_1,t_2},H^d\pi_{A})\in[\eta](\mathbf{o}_{t_1,t^\prime},H^d\pi_{A})$.
Third, if $\exists t\in [t_1,t^\prime], \eta(\mathbf{o}_{t},\pi_{A})<0 \wedge (t^\prime-t_1)<d$, using \eqref{eq:AGM_runtime_mntr2},
{\scriptsize $[\eta](\mathbf{o}_{t_1,t^\prime},H^d\pi_{A})= \left[ \frac{1}{d+1}\left(\sum\limits_{t\in[t_1,t^\prime]}[\eta(\mathbf{o}_t,\pi_A)]_\_+|[t^\prime,d]|\eta_{\mathrm{min}}\right) ,\frac{1}{d+1}\sum\limits_{t\in[t_1,t^\prime]}[\eta(\mathbf{o}_t,\pi_A)]_\_ \right]$}.
Then by Definition \ref{def:AGM_robustness}, $\eta(\mathbf{o}_{t_1,t_2},H^d\pi_{A})\in[\eta](\mathbf{o}_{t_1,t^\prime},H^d\pi_{A})$ which is a contradiction. 
Hence, $\eta(\mathbf{o}_{t_1,t_2},H^d\pi_{A})\in[\eta](\mathbf{o}_{t_1,t^\prime},H^d\pi_{A})$.

Then we have the following induction cases:

\noindent \textit{Within:} Follows similarly to the \textit{hold} case.

\noindent \textit{Conjunction/Disjunction:} The soundness follows by the properties of the interval arithmetic \ref{eq:AGM_interval_Arithm}.

\noindent \textit{Concatenation:} Given the Boolean semantic (\ref{eq:disjConjunCat}), the soundness of follows directly from the soundness of the \textit{Conjunction/Disjunction} case.
\end{proof}
\begin{lemma}(Convergence of Robustness Intervals)\label{lemma:MonitrosConvergance}
    Given a TWTL formula $\phi$ and word $\mathbf{o}_{t_1,t_2}$, where $t_2\geq||\phi||$; and for partial words $\mathbf{o}_{t_1,t_1^\prime}$ and $\mathbf{o}_{t_1,t_2^\prime}$, where $t_1^\prime<t_2^\prime<t_2$, the following set inclusions hold: $[\rho](\mathbf{o}_{t_1,t_1^\prime})\subseteq[\rho](\mathbf{o}_{t_1,t_2^\prime})$ and $[\eta](\mathbf{o}_{t_1,t_1^\prime})\subseteq[\eta](\mathbf{o}_{t_1,t_2^\prime})$. For $||\phi||\leq t^\prime\leq t_2$ the robustness intervals converge to a singleton which is the true robustness values, i.e.,
    $[\rho](\mathbf{o}_{t_1,t^\prime})= \{\rho(\mathbf{o}_{t_1,t_2})\}$ and
    $[\eta](\mathbf{o}_{t_1,t^\prime})= \{\eta(\mathbf{o}_{t_1,t_2})\}$.
\end{lemma}
\begin{example}(Continued)\label{eg:monitors}
    Given the TWTL unnormalized formula $\phi= [H^6 \pi_A]^{[0,10]}$ with time horizon $||\phi||=10$, we demonstrate the proposed runtime monitors by observing the robustness intervals $[\rho]$ and $[\eta]$ of partial words of $\mathbf{o}_1$, $\mathbf{o}_2$, and $\mathbf{o}_3$ (see the top-left figure of Fig.\ref{fig:allmntringrhoResults}). Consider the time series $\tau = \{0,\dots,10\}$.
    For each partial word we compute $[\rho]$ and $[\eta]$ at every $t\in\tau$ as we the partial words $\mathbf{o}_1(t)$, $\mathbf{o}_2(t)$, and $\mathbf{o}_3(t)$ become available.
    The evolution of the intervals $[\rho]$ and $[\eta]$ for $\mathbf{o}_1$, $\mathbf{o}_2$, and $\mathbf{o}_3$ are depicted in the top-right, bottom-left, and bottom-right figures of Fig. \ref{fig:allmntringrhoResults}, respectively, where the evolution of $[\rho]$ is depicted in dashed-lines with triangles and the evolution of $[\eta]$ is depicted in solid-lines with circles. Notice how the intervals become tighter as the partial word grows, where eventually when $t^\prime=||\phi||$, $[\rho]$ and $[\eta]$ converge to the true $\rho$ and $\eta$ values, respectively. In this example we consider $\rho_\top=10$ and $\rho_\bot=-10$; one can notice how the evolution of $[\eta]$ is smoother than the evolution of $[\rho]$, due to using the AGM in the computation of $\eta$.
     Note that we normalize the TWTL formula before computing $\eta$ robustness values.
    Thus, in Fig.~\ref{fig:allmntringrhoResults}, $\eta$ stays within $[-1, 1]$.
\end{example}

\begin{figure}\centering
\includegraphics[width=0.95\linewidth]{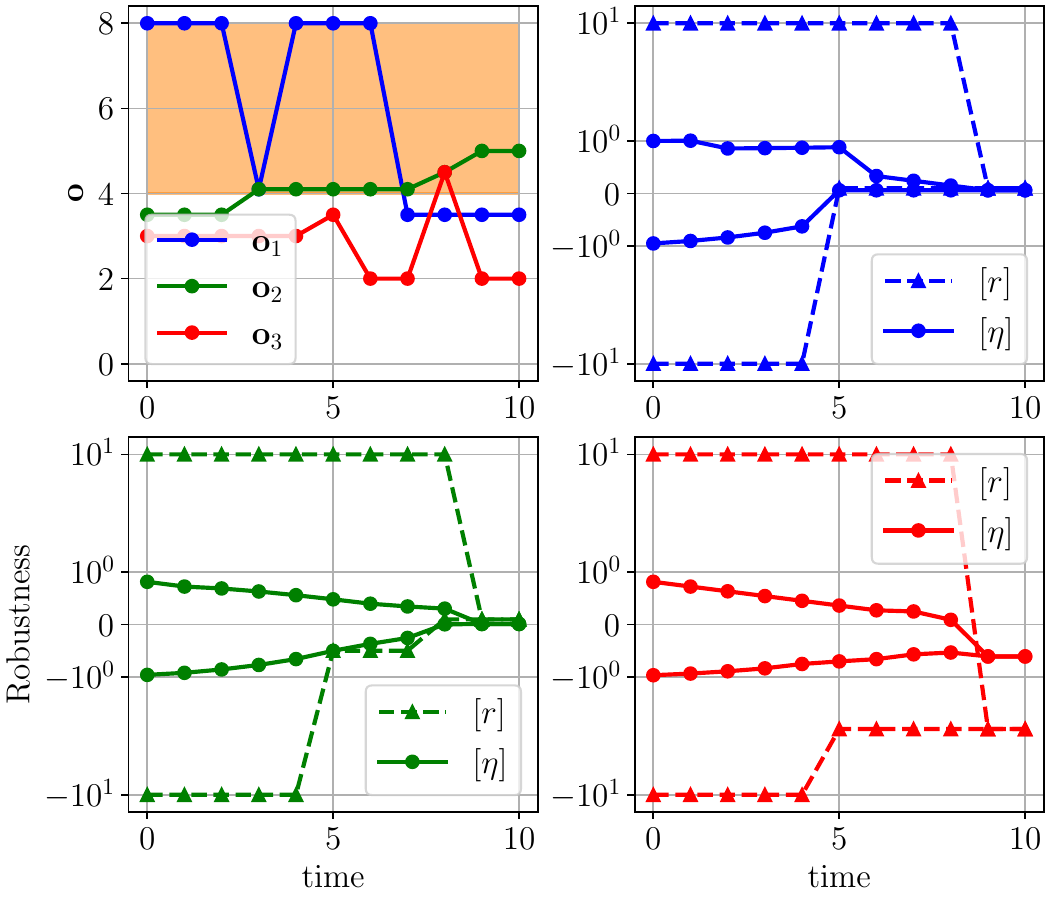}
	\caption{Demonstration of TWTL robustness and monitoring. (top-left) Depiction of the valuation of words $\mathbf{o}_1$, $\mathbf{o}_2$, and $\mathbf{o}_3$ w.r.t. time. The evolution of $[\rho]$ (dashed-line with triangles) and $[\eta]$ (solid-lines with circles) for words $\mathbf{o}_1$, $\mathbf{o}_2$ (bottom-left), and $\mathbf{o}_3$ are shown in the top-right, bottom-left, and bottom-right figures, respectively.} 
	\label{fig:allmntringrhoResults}
\end{figure}

\section{Numerical Example Case Study}\label{sec:example}
We demonstrate the proposed robustness semantics, by monitoring $\rho$ and $\eta$ for pre-computed runs of a simple planar robot system. Assume the time step, $\Delta t$, of the runs is $1$. We consider a simple sequential navigation task with deadlines and a safety requirement. In the following unnormalized TWTL formula, we encode the task that reads: \textit{Within time $0$ and $8$, visit region $A$ and stay there for $3$ time steps; right after that, within time $0$ and $10$, visit region $B$ and stay there for $4$ time steps; right after that, within time $0$ and $11$, visit region $C$ and stay there for $3$ time steps; and for all execution time avoid region $O$}. See the left figure of Fig. \ref{fig:caseStudy} for a depiction of the planar regions $A,B,C,$ and $O$. 
\begin{equation}\label{eq:phi_example}
    \phi = \left([H^4 \pi_A]^{[0,8]}\cdot[H^4 \pi_B]^{[0,10]}\cdot[H^3 \pi_C]^{[0,11]}\right)\wedge H^{50}\neg \pi_O 
\end{equation}
The atomic propositions $\pi_A,\pi_B,\pi_C$, and $\pi_O$ are defined as predicated regions over the $xy-\mathrm{plane}$; where $A:=\{(x,y)| 1\leq x\leq4 \wedge 1\leq y\leq4 \}$,
$B:=\{(x,y)| 8\leq x\leq11 \wedge 3\leq y\leq \}$,
$C:=\{(x,y)| 1\leq x\leq4 \wedge 9\leq y\leq12 \}$,
and $O:=\{(x,y)| 5\leq x\leq7 \wedge 5\leq x\leq7 \}$. 

We monitor two runs of the robot, $\mathbf{o}_1$ and $\mathbf{o}_2$, which are shown as the blue and green traces in the left figure of Fig. \ref{fig:caseStudy}, respectively. The robustness of $\mathbf{o}_1$ and $\mathbf{o}_2$ are $\rho(\mathbf{o}_1,\phi)=0.4$ and $\rho(\mathbf{o}_2,\phi)=0.3$, whereas their AGM robustness are $\eta(\mathbf{o}_1,\phi)=0.00076$ and $\eta(\mathbf{o}_2,\phi)=0.00015$. To monitor the robustness measures at runtime, consider the time series $\tau = \{2,10,15,20,25,30,35,40,42\}$.
For each partial word we compute $[\rho]$ and $[\eta]$ at every $t\in\tau$ as the partial words $\mathbf{o}_1(t)$ and $\mathbf{o}_2(t)$ become available. 
The valuations of the intervals $[\rho]$ and $[\eta]$ for $\mathbf{o}_1$, and $\mathbf{o}_2$ at every $t\in\tau$ are depicted in the middle, 
 and right figures of Fig. \ref{fig:caseStudy}, respectively, where $[\rho]$ is depicted in dashed-lines with triangles and the $[\eta]$ is depicted in solid-lines with circles. 

Observe that monitoring convergence of $[\eta]$ is smoother than the convergence of $[\rho]$, which would be more useful in some applications.
In our future work, we consider incremental planning for TWTL tasks for which we require runtime monitors to use as a heuristic to maximize the satisfaction of the task.
We find that aiming to maximize $\eta$ in planning applications would yield smoother paths, we leave the details of monitoring for planning applications for future work. 

\begin{figure*}\centering
\includegraphics[width=0.95\textwidth]{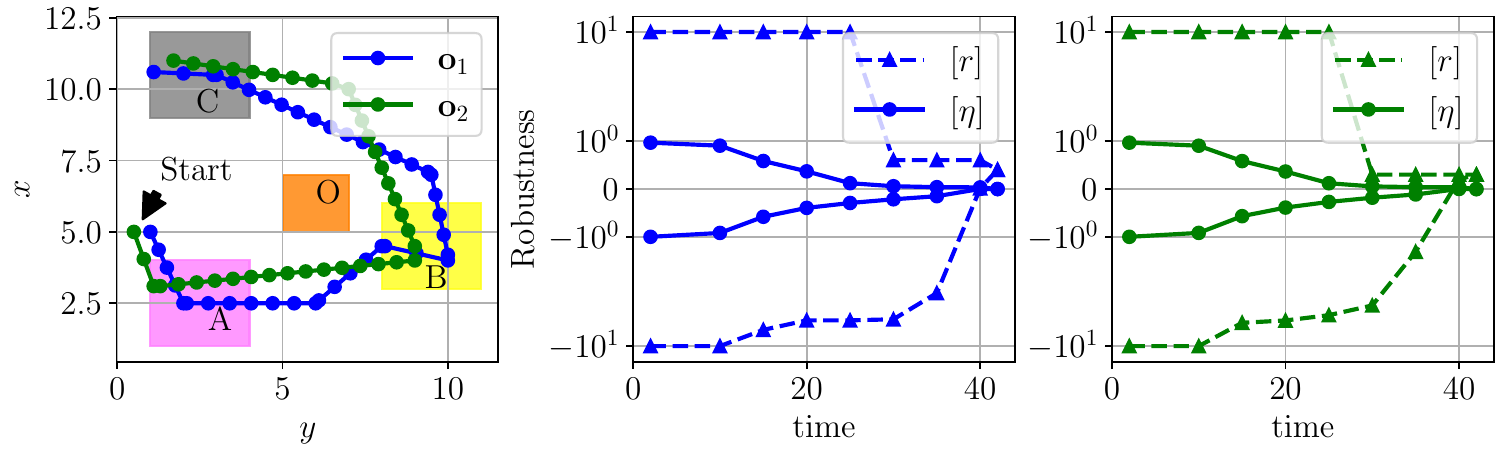}
	\caption{Demonstration of monitoring the TWTL robustness and AGM TWTL robustness of precomputed planar robot runs, $\mathbf{o}_1$ and $\mathbf{o}_1$, against satisfying formula (\ref{eq:phi_example}). Words $\mathbf{o}_1$ and $\mathbf{o}_2$ in an $xy-\mathrm{planar}$ environment are shown in the blue and green traces in the left figure, respectively. The valuations of the intervals $[\rho]$ and $[\eta]$ for $\mathbf{o}_1$, and $\mathbf{o}_2$ at every $t\in\{2,10,15,20,25,30,35,40,42\}$ are depicted in the middle, 
 and right figures, respectively, where $[\rho]$ is depicted in dashed-lines with triangles and the $[\eta]$ is depicted in solid-lines with circles. } 
	\label{fig:caseStudy}
\end{figure*}

\section{Conclusion and Future Work}\label{sec:conclusion}

Given the richness of Time Window Temporal Logic as a specification language for dynamical systems, we introduce two quantitative semantics to measure the robustness of TWTL formulae. In the first measure, which we call TWTL robustness, we introduce a distance measure between the system run and the formula satisfaction decision boundary based on the most critical values of the run. In the second measure, which we call AGM TWTL robustness, we quantify the satisfaction degree by another distance measure using the arithmetic and geometric mean of the system run values based on some rules that guarantee the soundness of the measure. In planning applications, AGM TWTL robustness enjoys a key advantage over the first one, in that it gives more reward to values that contribute to the formula satisfaction whereas TWTL robustness is dominated by the most critical values. We plan to demonstrate this advantage in future follow-up work. Given partial runs of the system, we develop runtime monitors that produce interval bounds on the quantitative semantics. We demonstrate the introduced quantitative semantics by monitoring TWTL robustness and AGM TWTL robustness of precomputed planar robot runs against some TWTL specifications.

\bibliographystyle{IEEEtran}
\bibliography{IEEEabrv,main_cdc}

\end{document}